\documentclass[pdflatex,sn-mathphys-num,iicol]{sn-jnl}

\usepackage{amsmath,amssymb,amsfonts}
\usepackage{float}
\usepackage{amsthm}

\theoremstyle{thmstyleone}
\newtheorem{theorem}{Theorem}[section]
\newtheorem{proposition}[theorem]{Proposition}

\usepackage{graphicx}
\usepackage{xcolor}
\usepackage{multirow}
\usepackage{array}
\usepackage{silence}
\usepackage{subcaption}
\usepackage{placeins} 
\usepackage{booktabs}

\usepackage{textcomp}
\usepackage{xurl}
\usepackage{microtype}          
\hypersetup{bookmarksdepth=2}

\usepackage{tikz}
\usetikzlibrary{positioning}

\title[Detectability Limits for Intra-Block Temporal Drift in Finite-Key QKD]{Detectability Limits for Intra-Block Temporal Drift in Finite-Key Entanglement-Based QKD}

\author[1]{\fnm{Rafael Duarte} \sur{Marcelino}}
\author[1,2]{\fnm{Julio Smanioto} \sur{Garcia}}
\author*[1]{\fnm{Matheus} \sur{Rufino}}\email{math.rufi@gmail.com}

\affil*[1]{\orgname{GWK}, \orgaddress{\city{Campinas}, \state{SP}, \country{Brazil}}}
\affil*[2]{\orgname{Instituto de Física Gleb Wataghin, Universidade Estadual de Campinas (UNICAMP)}, \orgaddress{\city{Campinas}, \state{SP}, \country{Brazil}}}

\abstract{We study the statistical detectability of intra-block temporal drift in finite-key entanglement-based quantum key distribution, with particular relevance to E91-type parameter estimation and monitoring. Drift is modeled as a mean-preserving Lipschitz perturbation of Bernoulli observables, capturing structured temporal variation that is invisible to global-average tests. For a block of size $n$ and confidence levels $(\alpha,\beta)$, we formulate a minimax hypothesis-testing problem and define the minimal detectable amplitude. We derive matching lower and upper bounds yielding $\delta_{\min}(n,\alpha,\beta)=\Theta(n^{-1/2})$: if $n\delta^2 \to 0$, no level-$\alpha$ procedure can guarantee nontrivial uniform power over the admissible drift class, whereas a calibrated CUSUM statistic detects drift at the matching scale. Explicit constants for linear, sinusoidal, and step profiles, together with simulations, confirm the predicted scaling collapse. The result quantifies a finite-block monitoring-resolution limit and is distinct from composable security certification.}


\begin{document}

\maketitle


\section{Introduction}
\label{sec:introduction}

In entanglement-based QKD systems, the physical channel can fluctuate within a
single parameter-estimation window in ways that cancel on average. Polarization
drift, birefringence fluctuations, Raman noise, and detector
instabilities~\cite{Eraerds2010,Mao2018,Xu2020Review} can therefore produce
intra-block non-stationarity that is invisible to diagnostics based only on the
block-average error rate. Table~\ref{tab:physical_profiles} summarizes how
common physical mechanisms map qualitatively onto the canonical drift
profiles studied in this paper.

\begin{table}[t]
\centering
\captionsetup{justification=raggedright,singlelinecheck=false}
\caption{Qualitative mapping between physical drift mechanisms in
entanglement-based QKD and the canonical mean-preserving profiles studied
here. Precise calibration is system-dependent.}
\label{tab:physical_profiles}
\small
\setlength{\tabcolsep}{4pt}
\begin{tabular}{@{}>{\raggedright\arraybackslash}p{0.42\columnwidth}ll@{}}
\toprule
Physical mechanism & Scale & Profile \\
\midrule
Polarization / birefringence drift & ms--s & sinusoidal \\
Thermal / environmental drift & s--min & linear \\
Abrupt detector or alignment events & burst & step \\
\bottomrule
\end{tabular}
\end{table}

Finite-key analysis in quantum key distribution (QKD) quantifies the
secret-key length that can be certified from a finite number of signals
under composable security requirements. For entanglement-based protocols
such as E91~\cite{Ekert1991}, the relevant observables include the quantum
bit error rate (QBER) and, in device-independent settings, Bell-violation
statistics. Modern finite-size security methods, including smooth-entropy
techniques and the Entropy Accumulation Theorem (EAT), provide rigorous
security reductions beyond the i.i.d.\ regime~\cite{Devetak2005,Renner2008,Tomamichel2012,Dupuis2020}.
These advances establish security certification, but they do not by
themselves quantify monitoring resolution for temporal structure within a
single parameter-estimation block.

Although the detectability boundary for mean-preserving intra-block
drift is under-characterized in finite blocks, this regime is precisely
where global tests on $\hat e$ become uninformative: when temporal
variation integrates to zero, the problem reduces to structured signal
detection under finite-sample noise.
Figure~\ref{fig:erro_oscilante} illustrates a trajectory $e(t)$ whose
block average matches the baseline $e_0$ exactly, even though its temporal
structure deviates substantially from stationarity.

\begin{figure}[t]
\centering
\includegraphics[width=\columnwidth]{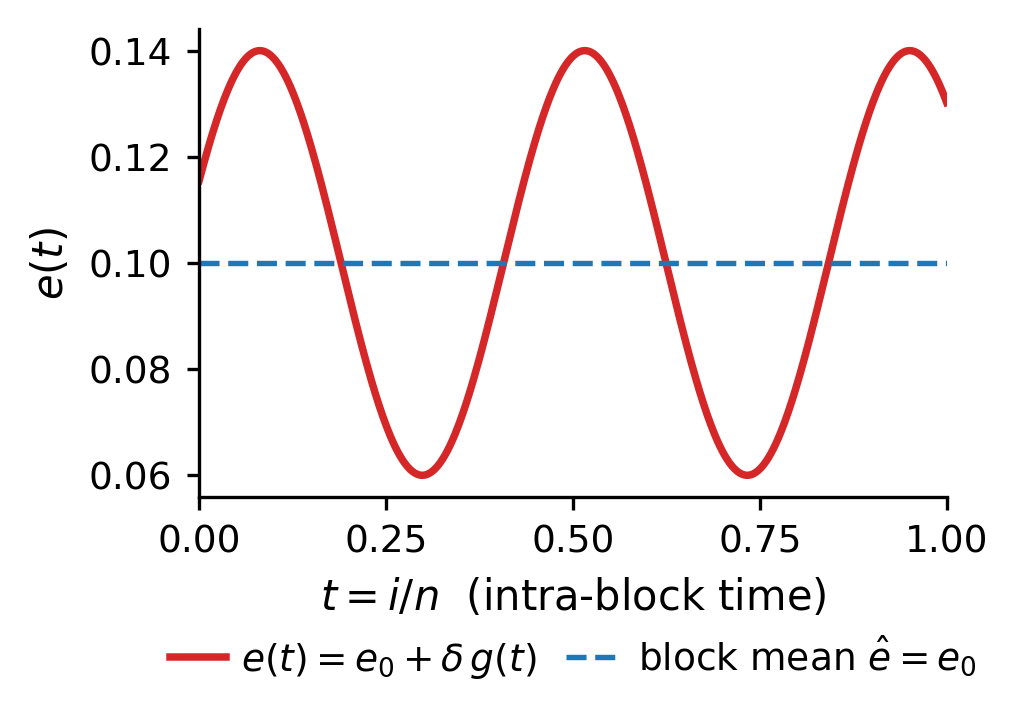}
\caption{Illustration of mean-preserving intra-block drift. The error
probability $e(t)$ oscillates within the block while its integral over the
block remains equal to the baseline $e_0$. A monitor that inspects only the
block average $\hat e$ cannot distinguish this trajectory from a stationary
one, motivating the partial-sum approach developed in Section~\ref{sec:methods}.}
\label{fig:erro_oscilante}
\end{figure}

\FloatBarrier%

Methods based on the Entropy Accumulation Theorem~\cite{Dupuis2020} and smooth
min-entropy~\cite{Renner2008,Tomamichel2012} certify composable security
\emph{conditional on} the device model remaining valid throughout the
estimation block. The present work addresses a complementary and prior
question: given finite data from a single block, can one detect that this
stationarity assumption has been violated? This is a physical diagnostic
problem, not a security reduction.

This paper addresses that gap by formulating intra-block drift detection as a
minimax hypothesis-testing problem for blockwise Bernoulli observations under
mean-preserving alternatives. The main question is: for block size $n$ and
error targets $(\alpha,\beta)$, what is the smallest amplitude
\[
\delta_{\min}(n,\alpha,\beta)
\]
that can be detected with controlled type-I and type-II errors?

Our contributions are threefold. First, we derive matching lower and upper
bounds showing
\[
\delta_{\min}(n,\alpha,\beta)=\Theta(n^{-1/2}).
\]
Second, we show that a calibrated plug-in CUSUM statistic~\cite{Page1954} attains this rate
over admissible drift classes. Third, we compute explicit signal constants for
linear, sinusoidal, and step profiles and validate the predicted scaling
through simulations.

The remainder of the paper is organized as follows.
Section~\ref{sec:methods} presents the statistical model, testing framework,
detectability analysis, and canonical drift profiles. Section~\ref{sec:results_discussion}
reports simulation evidence and operational implications. Section~\ref{sec:conclusion}
summarizes the conclusions.

\section{Methods}
\label{sec:methods}

The analysis proceeds in three stages. We first define the statistical model
and the class of mean-preserving drift alternatives that are invisible to
global-average diagnostics. We then show why mean-based statistics are
structurally blind to these alternatives, derive the cumulative partial-sum
instrument that captures the signal the average discards, and establish
matching lower and upper bounds on the minimal detectable amplitude. Finally,
we instantiate the universal $\Theta(n^{-1/2})$ rate on three canonical drift
profiles, computing the signal constant $A(g)$ that governs detection
sensitivity for each.

\subsection{Statistical Model}
\label{sec:statistical_model}

We consider a finite-key parameter-estimation block of size $n$ in an
entanglement-based QKD protocol. Let $X_i \in \{0,1\}$ denote the key-basis
error indicator at trial $i$, with $X_i=1$ for disagreement and $X_i=0$
otherwise. We associate each trial with the normalized intra-block time
coordinate
\[
t_i = \frac{i}{n}, \qquad i = 1,\dots,n.
\]

In E91-type implementations, the relevant error statistics depend on the
relative polarization alignment between remote analyzers. Let $\theta(t)$
denote a slowly varying misalignment parameter induced by environmental or
hardware fluctuations. For small perturbations around an operating point, a
first-order expansion yields
\[
e(t) = e_0 + \delta g(t),
\]
where $e_0 \in (0,1)$ is a baseline error probability, $\delta \ge 0$ is the
drift amplitude, and $g:[0,1]\to\mathbb{R}$ encodes the normalized temporal
profile. This Bernoulli abstraction captures the leading-order statistical
effect of polarization drift without requiring a detailed optical model.

We therefore observe a temporal sequence of binary errors and ask whether that
sequence is stationary or whether it contains structured intra-block variation
that the average error rate conceals. The hypotheses below formalize this
distinction between a constant error process and a mean-preserving drift
alternative.

\paragraph{Null hypothesis: stationarity.}

Under the stationary hypothesis $H_0$, observations are
independent Bernoulli trials with constant error probability:
\[
H_0:\qquad
X_i \sim \mathrm{Bernoulli}(e_0),
\quad i=1,\dots,n,
\]
for some fixed $e_0 \in (0,1)$.

The block-average estimator is
\[
\hat e = \frac{1}{n} \sum_{i=1}^n X_i.
\]

\paragraph{Alternative: mean-preserving drift.}

Under the alternative hypothesis $H_1$, the error probability
varies within the block according to
\[
e(t) = e_0 + \delta g(t),
\]
where:

\begin{itemize}
\item $\delta \ge 0$ is a drift amplitude parameter,
\item $g:[0,1]\to\mathbb{R}$ is a deterministic drift profile,
\item $\int_0^1 g(t)\,dt = 0$ (mean preservation),
\item $0 \le e_0 + \delta g(t) \le 1$ for all $t\in[0,1]$.
\end{itemize}

For later normalization, define the baseline variance
\[
\sigma_0^2 := e_0(1-e_0).
\]
This quantity governs the fluctuation scale of partial-sum processes
under stationarity.

Conditionally on $g$, observations are independent:
\[
\begin{aligned}
H_1:\qquad
X_i &\sim \mathrm{Bernoulli}\big(e_0 + \delta g(t_i)\big),\\
&\qquad i=1,\dots,n.
\end{aligned}
\]

Because of this constraint, $\mathbb{E}_{H_1}[\hat e] = e_0$: the
block-average error rate is identical under $H_0$ and $H_1$. Any test that
relies solely on $\hat e$ is therefore blind to the alternative by
construction.

The mean-preserving constraint models an operationally relevant regime in
which oscillatory or slowly varying perturbations average out over the block
and evade mean-based abort rules while still inducing systematic temporal
structure.

\paragraph{Drift class.}

To obtain non-degenerate detectability statements, we work with the class of
admissible drift functions
\[
\mathcal{G}(L,c)
=
\left\{
g :
\begin{array}{l}
\int_0^1 g(t)\,dt = 0, \\
\|g\|_\infty \le 1, \\
\mathrm{Lip}(g) \le L, \\
\|g\|_2 \ge c
\end{array}
\right\},
\]
for constants $L>0$ and $c>0$. The Lipschitz bound is a physical regularity
condition: oscillations faster than the monitoring resolution are
unobservable and should not determine the hardest-case alternative. The lower
bound on $\|g\|_2$ excludes drift profiles with negligible energy, which are
statistically indistinguishable from the null regardless of amplitude.
Together, these constraints ensure the detectability problem is non-degenerate.

With the model and admissible alternatives in place, the central question is:
which statistics can detect mean-preserving drift, and at what amplitude does
detection become possible?

\subsection{CUSUM Test and Detectability}
\label{sec:testing_detectability}

\paragraph{The fundamental difficulty.}

The mean-preserving constraint $\int_0^1 g(t)\,dt = 0$ forces
$\mathbb{E}_{H_1}[\hat e] = e_0$: the block-average QBER is identical under
$H_0$ and $H_1$. Any test relying solely on $\hat e$ is therefore structurally
blind to the alternative. This is not a limitation of a particular test design
but a consequence of the model: the global sum $\sum_{i=1}^n (X_i - e_0)$
carries no information about $\delta$ when $g$ is mean-preserving.

The signal is nevertheless present in the data; it is distributed across time
rather than stored in the total. Under $H_1$, early observations can lie
systematically above baseline while later ones compensate, producing a
detectable excursion in the cumulative record that cancels only at the very
end of the block. This motivates monitoring partial sums rather than the
global average.

\paragraph{Partial-sum process and signal analysis.}

Define centered variables
\[
Y_i := X_i - \hat e,
\]
and the associated partial-sum process
\begin{equation}
S_k := \sum_{i=1}^k Y_i,
\qquad k=1,\dots,n.
\label{eq:Sk_def}
\end{equation}

Under $H_0$, this process is bridge-like: $S_n = 0$ by construction, and
after normalization by $\sigma_0\sqrt{n}$ it behaves asymptotically as a
Brownian bridge, with fluctuations of order $O(1)$.

Under $H_1(\delta,g)$, with $t_i=i/n$ and $p_i=e_0+\delta g(t_i)$, we have
$\mathbb{E}_{H_1}[X_i]=p_i$ and $\mathbb{E}_{H_1}[\hat e]=e_0$, so
\begin{align}
\mathbb{E}_{H_1}[S_k]
&=
\sum_{i=1}^k \mathbb{E}_{H_1}[X_i-\hat e]
\approx
\delta \sum_{i=1}^k g(t_i)
\notag\\
&\approx
\delta n \int_0^{k/n} g(u)\,du.
\label{eq:Sk_mean_alt}
\end{align}
Thus $S_k$ acquires a deterministic drift whose shape is governed by the
cumulative integral of $g$. Define the cumulative drift profile
\begin{equation}
G(t) := \int_0^t g(u)\,du,
\qquad t\in[0,1],
\label{eq:G_def}
\end{equation}
and the signal functional
\begin{equation}
A(g) := \sup_{t\in[0,1]} |G(t)|.
\label{eq:Ag_def}
\end{equation}
From~\eqref{eq:Sk_mean_alt}, the maximum deterministic excursion of $S_k$
scales as $\delta n \cdot A(g)$, while null fluctuations scale as
$\sigma_0\sqrt{n}$. The effective signal-to-noise ratio is therefore
\[
\mathrm{SNR} \asymp \frac{\delta\sqrt{n}\,A(g)}{\sigma_0},
\]
which is of order one precisely when $\delta \asymp n^{-1/2}$. This is the
natural critical scale; Figure~\ref{fig:cusum_trajectory} contrasts
the partial-sum process under $H_0$ and $H_1$.

\begin{figure}[t]
\centering
\includegraphics[width=\columnwidth]{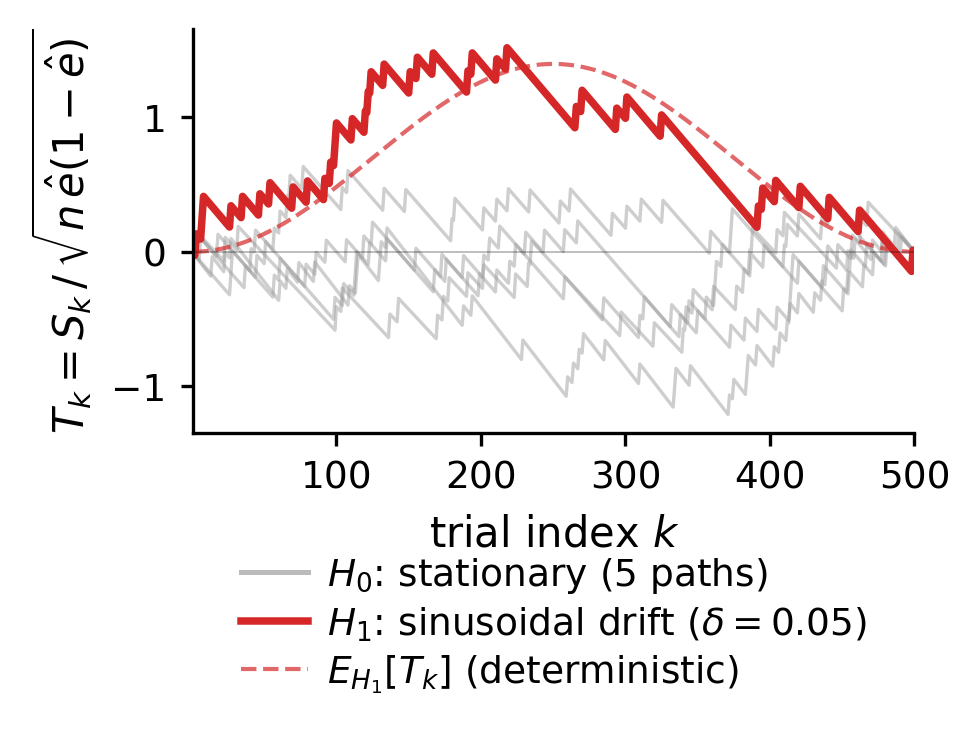}
\caption{Partial-sum process $S_k$ under $H_0$ (grey) and under a
mean-preserving sinusoidal alternative $H_1$ (red). Under $H_0$ the process
fluctuates as a bridge with $S_n=0$; under $H_1$ it develops a systematic
excursion of order $\delta n A(g)$ before returning to zero. The CUSUM
statistic $T_n$ captures the maximum of $|S_k|$, which separates the two
distributions once $\delta\sqrt{n}\,A(g)$ exceeds the null fluctuation scale.}
\label{fig:cusum_trajectory}
\end{figure}

\FloatBarrier%

\paragraph{CUSUM statistic.}

The signal analysis leads directly to the test statistic. Because $e_0$ is
unknown at monitoring time, we substitute $\hat e$ in both the centering and
the normalization:
\[
T_n :=
\max_{1 \le k \le n}
\left|
\frac{\sum_{i=1}^k (X_i - \hat e)}
{\sqrt{n \hat e (1-\hat e)}}
\right|.
\]
Under $H_0$, $T_n$ converges in distribution to the supremum of a Brownian
bridge, permitting calibration from the asymptotic distribution or Monte Carlo
samples. Under $H_1(\delta,g)$, the deterministic component of $T_n$ scales
as $\delta\sqrt{n}\,A(g)$, while null fluctuations remain $O(1)$.

\paragraph{Testing framework and detectability.}

A (possibly randomized) test is a measurable function
\[
\varphi_n:\{0,1\}^n \to [0,1],
\]
where $\varphi_n(x)$ is the probability of rejecting $H_0$ upon observing $x$.
We use the standard type-I/type-II error definitions:

\begin{align}
\alpha(\varphi_n)
&:= \sup_{e_0\in(0,1)} \ \mathbb{P}_{H_0}\!\left(\varphi_n=1\right),
\label{eq:typeI}\\
\beta(\varphi_n;\delta)
&:= \sup_{e_0\in(0,1)} \ \sup_{g\in\mathcal{G}}
\mathbb{P}_{H_1(\delta,g)}\!\left(\varphi_n=0\right).
\label{eq:typeII}
\end{align}

Fix confidence levels $(\alpha,\beta)\in(0,1)^2$. The minimal detectable
amplitude, subject to the error definitions \eqref{eq:typeI}--\eqref{eq:typeII}, is
\begin{equation}
\begin{aligned}
\delta_{\min}(n,\alpha,\beta) &:= \inf\Big\{\delta>0 \;:\; \exists \varphi_n\\
&\enspace \text{with } \alpha(\varphi_n)\le\alpha \text{ and } \beta(\varphi_n;\delta)\le\beta\Big\}.
\end{aligned}
\label{eq:delta_min_def}
\end{equation}

This is the smallest drift amplitude for which one can discriminate $H_1$
from $H_0$ on a block of length $n$ while controlling false alarms uniformly
over $\mathcal{G}(L,c)$. The SNR analysis above suggests
$\delta_{\min}(n,\alpha,\beta) = \Theta(n^{-1/2})$; the following two results
confirm this at both sides of the critical scale.

\paragraph{Upper bound: CUSUM achievability.}\label{sec:constructive_test}

Let $\tau_\alpha(n)$ denote a level-$\alpha$ calibration threshold for $T_n$
under $H_0$, obtained either from the asymptotic bridge distribution or from
finite-sample Monte Carlo calibration. The corresponding test is
\begin{equation}
\varphi_n := \mathbf{1}\{T_n > \tau_\alpha(n)\}.
\label{eq:test_def}
\end{equation}

\begin{proposition}[CUSUM upper bound]
\label{prop:cusum_upper}
Assume the admissible subclass under consideration satisfies
$A(g) \ge A_\star > 0$ for all $g$ in the class. Then, for fixed
$(\alpha,\beta) \in (0,1)^2$, there exists a constant
$C_{\alpha,\beta}(e_0,A_\star)$ such that the test \eqref{eq:test_def}
has power at least $1-\beta$ for all sufficiently large $n$ whenever
\[
\delta \ge \frac{C_{\alpha,\beta}(e_0,A_\star)}{\sqrt{n}}.
\]
\end{proposition}

Proof in Appendix~\ref{app:cusum_upper_proof}.

Once the deterministic excursion $\delta\sqrt{n}A(g)$ exceeds the null
fluctuation scale, the CUSUM statistic separates $H_1$ from the bridge-like
null process.

\paragraph{Lower bound: fundamental limitation.}

The SNR analysis identifies $\delta \asymp n^{-1/2}$ as the critical scale but
does not rule out a smarter test achieving better performance. The following
result shows that no such test exists.

\begin{theorem}[Lower bound for uniform detectability]
\label{thm:minimax_lower}
Let $\mathcal{G}(L,c)$ be the drift class introduced in
Section~\ref{sec:statistical_model}. Fix $\alpha \in (0,1)$ and a sequence
$\delta_n$ satisfying $n\delta_n^2 \to 0$.

Then, for any sequence of tests $(\varphi_n)$ with
$\alpha(\varphi_n) \le \alpha$, one has
\[
\inf_{g \in \mathcal{G}(L,c)}
\mathbb{P}_{H_1(\delta_n,g)}\!\left(\varphi_n=1\right)
\le \alpha + o(1).
\]

In particular, no level-$\alpha$ procedure can achieve power uniformly bounded
away from $\alpha$ over $\mathcal{G}(L,c)$ when $\delta_n=o(n^{-1/2})$.
\end{theorem}

Proof in Appendix~\ref{app:lower_bound_proof}.

Together, Theorem~\ref{thm:minimax_lower} and Proposition~\ref{prop:cusum_upper}
establish that CUSUM is rate-optimal.

\begin{theorem}[Power scaling of the detectable amplitude]
\label{thm:power_asymptotics}
Under the regularity conditions of Theorem~\ref{thm:minimax_lower} and
Proposition~\ref{prop:cusum_upper}, and for fixed
$(\alpha,\beta)\in(0,1)^2$,
\[
\delta_{\min}(n,\alpha,\beta)
=
\Theta\!\left(n^{-1/2}\right).
\]
The CUSUM statistic attains the minimax lower bound up to a constant; no
level-$\alpha$ test can achieve nontrivial power $1{-}\beta$ at amplitudes
$\delta$ with $n\delta^2\to 0$.
\end{theorem}

\subsection{Canonical Drift Profiles}
\label{sec:canonical_drift_models}

The rate $\Theta(n^{-1/2})$ is universal across the admissible drift class,
but the prefactor depends on the signal functional $A(g)$, which varies with
the shape of the drift. To make this operational, we compute $A(g)$
explicitly for three physically motivated profiles and derive the quantitative
ordering of CUSUM sensitivity that the simulations will verify.

Throughout, we enforce the normalization
\begin{equation}
\|g\|_\infty = 1,
\label{eq:g_norm_infty}
\end{equation}
so that the amplitude parameter $\delta$ has the same meaning across models.
We also impose the mean-preserving constraint
$\int_0^1 g(t)\,dt=0$.

The three normalized profiles used in the simulations are
\[
\begin{aligned}
g_{\mathrm{lin}}(t) &= 2t-1,\\
g_{\mathrm{sin}}(t) &= \sin(2\pi t),\\
g_{\mathrm{step}}(t) &=
\begin{cases}
-1,&0\le t<\tfrac{1}{2},\\
+1,&\tfrac{1}{2}\le t\le 1.
\end{cases}
\end{aligned}
\]

\paragraph{Predicted ordering.}

The relevant model-dependent quantity is the CUSUM signal constant
\[
A(g) = \sup_{t\in[0,1]} \left| \int_0^t g(u)\,du \right|.
\]
The signal functional $A(g)$ measures the maximum cumulative imbalance that the
drift profile can produce within a block. A profile with larger $A(g)$ creates
a more visible excursion in the partial sums, and is therefore easier to
detect at fixed $(n, \delta)$.
For the three normalized profiles used in the simulations, the constants are
\begin{equation}
\begin{aligned}
A(g_{\mathrm{lin}}) &= \frac{1}{4},\\
A(g_{\mathrm{sin}}) &= \frac{1}{\pi}\approx 0.318,\\
A(g_{\mathrm{step}}) &= \frac{1}{2}.
\end{aligned}
\label{eq:canonical_A_values}
\end{equation}
Derivations are given in Appendix~\ref{app:canonical_derivations}. Since
$1/4 < 1/\pi < 1/2$, the deterministic excursion induced by step drift is
largest, followed by sinusoidal drift, with linear drift producing the
smallest excursion. Figure~\ref{fig:constantes_sinal} plots the three
cumulative profiles $G(t)$ and the signal constant $A(g)$ for each.

\begin{figure}[t]
\centering
\includegraphics[width=\columnwidth]{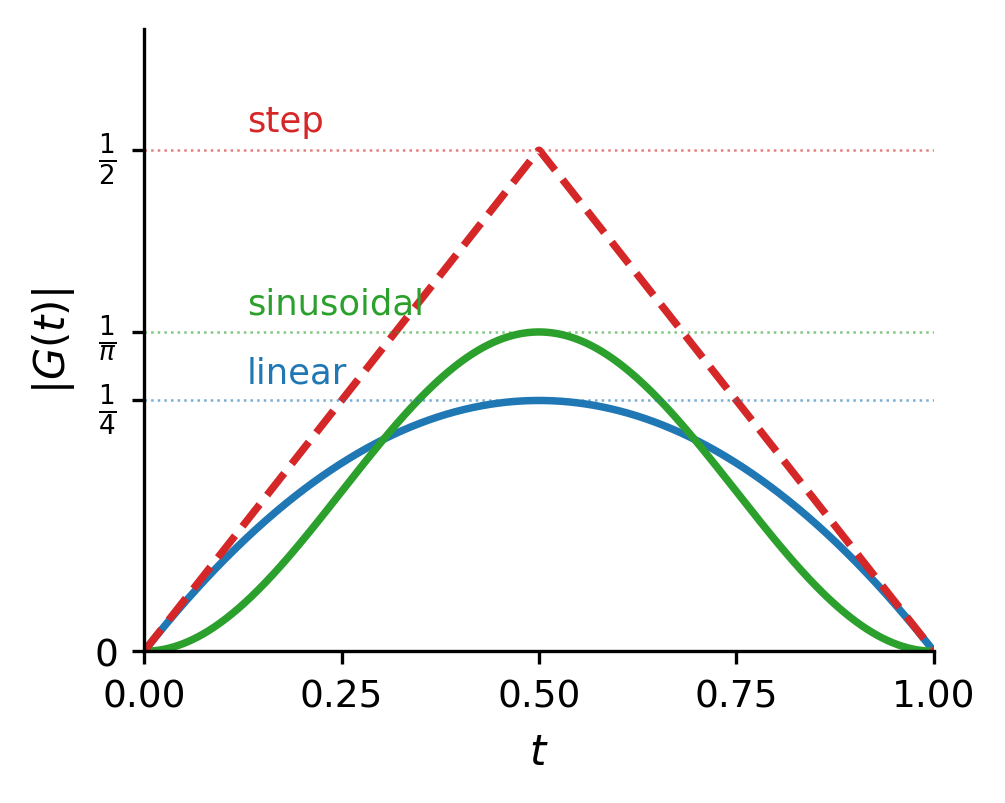}
\caption{Cumulative drift profiles $G(t)=\int_0^t g(u)\,du$ for the three
normalized drift families. The signal constant $A(g)=\sup_t|G(t)|$ equals the
maximum vertical excursion of each curve: $1/4$ for linear (blue), $1/\pi$
for sinusoidal (green), and $1/2$ for step (red). The ordering $A(g_{\mathrm{step}})
> A(g_{\mathrm{sin}}) > A(g_{\mathrm{lin}})$ directly predicts the sensitivity
ordering of the CUSUM test.}
\label{fig:constantes_sinal}
\end{figure}

\FloatBarrier%

Because detectability at fixed $(n,\delta)$ is governed by the effective signal-to-noise ratio
\[
\delta \sqrt{n}\, A(g),
\]
we predict the qualitative ordering

\[
\text{step} \;>\; \text{sinusoidal} \;>\; \text{linear}
\]

in terms of statistical sensitivity.

Equivalently, for fixed $(n,\alpha,\beta)$,
\[
\delta_{\min}^{\text{step}}
<
\delta_{\min}^{\text{sin}}
<
\delta_{\min}^{\text{lin}}.
\]

\subsection{Scope and Limitations}
\label{sec:scope_limitations}

The present analysis assumes conditional independence within each monitoring block. Dependent observation models are not considered here.

\clearpage
\section{Results and Discussion}
\label{sec:results_discussion}

\subsection{Simulation Setup}
\label{sec:simulation_framework}

The theoretical predictions of Section~\ref{sec:methods} make three testable
claims: the $n^{-1/2}$ scaling of $\delta_{\min}$, the scaling collapse
under the rescaled variable $\delta\sqrt{n}$, and the model-dependent
ordering of detectability
$\widehat{\delta}_{\min}^{\mathrm{step}}<
\widehat{\delta}_{\min}^{\mathrm{sin}}<
\widehat{\delta}_{\min}^{\mathrm{lin}}$.
This section verifies each claim numerically.

All simulations are performed under the model of
Section~\ref{sec:statistical_model} and the CUSUM-based procedure of
Section~\ref{sec:constructive_test}.

Five block sizes ($n = 250$, $500$, $1000$, $2000$, $4000$), three baseline
error rates ($e_0 = 0.02$, $0.05$, $0.10$), significance level $\alpha=0.05$,
and target power $1-\beta=0.8$ are used throughout. Three normalized drift
profiles are tested: linear, sinusoidal, and step
(Section~\ref{sec:statistical_model}).

For each $(n,e_0,g)$, admissibility requires
\[
0 \le e_0 + \delta g(t) \le 1
\quad \forall t\in[0,1],
\]
so the maximal amplitude is
\begin{equation}
\delta_{\max}=\min(e_0,\,1-e_0).
\label{eq:delta_max}
\end{equation}

For fixed $(n,e_0,g)$, define
\begin{multline}
\widehat{\delta}_{\min}(n,\alpha,\beta,g) :=\\
\inf\!\left\{\delta :
\widehat{\mathrm{Power}}(n,e_0,g,\delta) \ge 1-\beta\right\}.
\label{eq:delta_min_hat}
\end{multline}

\subsection{Power Curves and Scaling Collapse}
\label{sec:simulation_results}

The empirical results test both the leading-order scaling with
$\delta\sqrt{n}$ and the model-dependent refinement through $A(g)$. We analyze
linear, sinusoidal, and step drift profiles across baseline error rates
$e_0 \in \{0.02,0.05,0.10\}$.

Figure~\ref{fig:delta_min_scaling} summarizes the key empirical result:
$\widehat{\delta}_{\min}$ plotted against block size $n$ on logarithmic
scales for the three drift families ($e_0=0.05$, $\alpha=0.05$, target
power $1{-}\beta=0.8$). All three curves run parallel to the reference line
$\propto n^{-1/2}$, confirming the predicted scaling rate. The vertical
ordering---step below sinusoidal below linear---reflects the hierarchy
$A(g_{\mathrm{step}}) > A(g_{\mathrm{sin}}) > A(g_{\mathrm{lin}})$ stated
in Theorem~\ref{thm:power_asymptotics}.

\begin{figure}[t]
\centering
\includegraphics[width=\columnwidth]{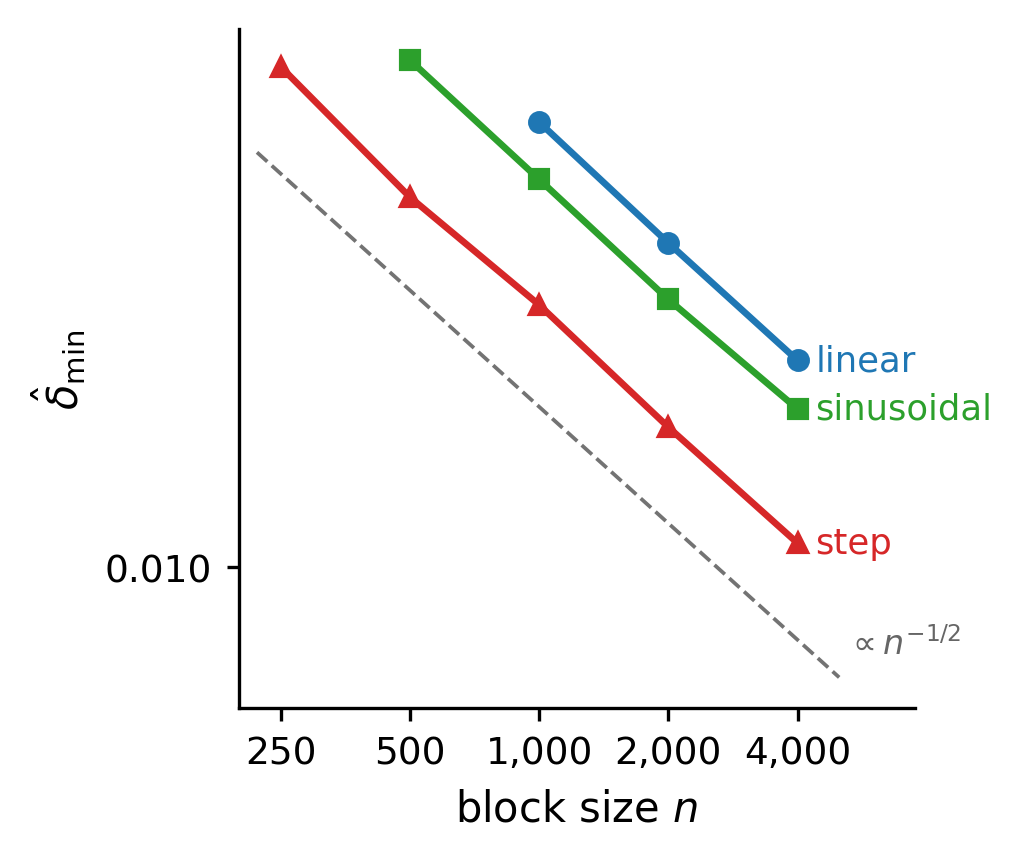}
\caption{Empirical detectability threshold $\widehat{\delta}_{\min}$ vs.\
block size $n$ (log--log scale) at baseline error rate $e_0=0.05$.
All three drift families---linear (blue circles), sinusoidal (green squares),
step (red triangles)---decay parallel to the dashed reference line
$\propto n^{-1/2}$, confirming Theorem~\ref{thm:power_asymptotics}.
The vertical separation encodes the model-dependent signal constant $A(g)$:
step drift concentrates energy at the block boundary, yielding the largest
$A(g)$ and hence the smallest $\widehat{\delta}_{\min}$.}
\label{fig:delta_min_scaling}
\end{figure}

\FloatBarrier%

We next test the rescaled variable $\delta\sqrt{n}$.

Power curves rescaled by $\delta\sqrt{n}$ collapse onto a single master
curve for each drift family across all block sizes $n\in\{250,\dots,4000\}$,
confirming that $\delta\sqrt{n}$ is the natural scaling variable. The
horizontal offset between families is consistent with the signal-constant
ratio $A(g_{\mathrm{step}})/A(g_{\mathrm{lin}})=2$ derived in
Section~\ref{sec:canonical_drift_models}. Residual spread at the transition
region reflects finite-sample effects that diminish as $n$ grows.

\begin{figure*}[t]
\centering
\includegraphics[width=\textwidth]{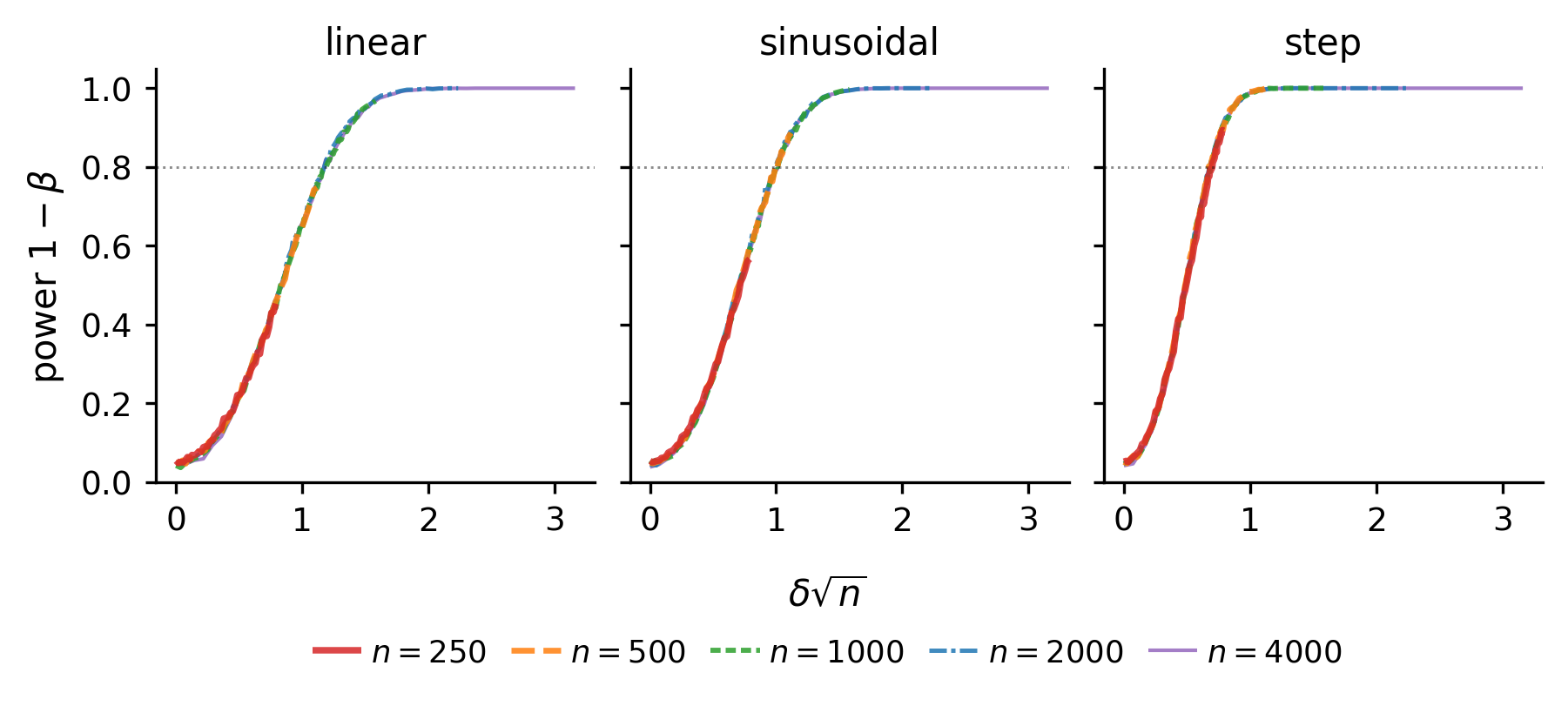}
\caption{Scaling collapse: power as a function of the rescaled amplitude
$\delta\sqrt{n}$ for $e_0=0.05$, shown separately for each drift family.
Each curve corresponds to a block size $n\in\{250,500,1000,2000,4000\}$,
distinguished by color and line style (see legend). Curves align closely
onto a single master curve in each panel, confirming that $\delta\sqrt{n}$
is the universal scaling variable. Residual spread near the transition
region, most visible for $n=250$ (red, plotted on top), reflects
finite-sample effects that diminish as $n$ grows. The horizontal shift
between panels reflects the model-dependent signal constant $A(g)$: the
step profile (right) transitions at smaller $\delta\sqrt{n}$ than the
linear profile (left), consistent with
$A(g_{\mathrm{step}})/A(g_{\mathrm{lin}})=2$.}
\label{fig:scaling_collapse}
\end{figure*}

\FloatBarrier%

Across all configurations, power increases monotonically with $n$, and the
transition from low to high power sharpens as the block size grows. The
amplitude required to achieve fixed power decreases approximately as
$n^{-1/2}$, consistent with the theoretical prediction.

Table~\ref{tab:delta_min_summary} summarizes this transition through the
empirical detectability threshold $\widehat{\delta}_{\min}$ defined in
\eqref{eq:delta_min_hat}. The table makes two aspects explicit: first,
for fixed $(e_0,g)$ the threshold decreases with $n$; second, for fixed
$(n,e_0)$ the ordering
\[
\widehat{\delta}_{\min}^{\mathrm{step}}
<
\widehat{\delta}_{\min}^{\mathrm{sin}}
<
\widehat{\delta}_{\min}^{\mathrm{lin}}
\]
appears whenever all three models attain the target power. Entries marked
``--'' correspond to configurations in which power $0.8$ was not reached
within the admissible interval $[0,\delta_{\max}]$.

\begin{table*}[t]
\centering
\caption{Empirical detectability threshold
$\widehat{\delta}_{\min}(n,\alpha,\beta,g)$ at target power $1-\beta=0.8$.
Each entry is the smallest simulated amplitude whose interpolated power curve
reaches the target.}
\label{tab:delta_min_summary}
\small
\setlength{\tabcolsep}{5pt}
\begin{tabular}{cccccccccc}
\toprule
\multirow{2}{*}{$n$}
& \multicolumn{3}{c}{$e_0=0.02$}
& \multicolumn{3}{c}{$e_0=0.05$}
& \multicolumn{3}{c}{$e_0=0.10$} \\
\cmidrule(lr){2-4} \cmidrule(lr){5-7} \cmidrule(lr){8-10}
& Linear & Sinusoidal & Step
& Linear & Sinusoidal & Step
& Linear & Sinusoidal & Step \\
\midrule
250
& -- & -- & --
& -- & -- & 0.0443
& -- & 0.0862 & 0.0590 \\
500
& -- & -- & --
& -- & 0.0452 & 0.0301
& 0.0719 & 0.0609 & 0.0417 \\
1000
& -- & -- & 0.0140
& 0.0376 & 0.0317 & 0.0218
& 0.0507 & 0.0430 & 0.0292 \\
2000
& 0.0171 & 0.0143 & 0.0098
& 0.0263 & 0.0222 & 0.0152
& 0.0353 & 0.0298 & 0.0206 \\
4000
& 0.0118 & 0.0099 & 0.0069
& 0.0185 & 0.0160 & 0.0107
& 0.0252 & 0.0215 & 0.0147 \\
\bottomrule
\end{tabular}
\end{table*}

Figure~\ref{fig:scaling_collapse} shows the scaling collapse for each drift
family. Across all drift families and baseline error rates, curves corresponding to
different block sizes align closely when plotted against $\delta\sqrt{n}$.
The ordering across drift shapes is consistent with the refinement through
$A(g)$ computed in Section~\ref{sec:canonical_drift_models}.

Although the scaling exponent $1/2$ is universal across drift families,
the transition region from low to high power differs systematically
between models.

Consistent with the signal functional $A(g)$,
the empirical ordering of sensitivity is:

\begin{itemize}
\item \textbf{Step drift} exhibits the sharpest transition and achieves
high power at the smallest amplitudes, reflecting its maximal cumulative imbalance.

\item \textbf{Sinusoidal drift} requires slightly larger amplitude than
step drift but remains more detectable than linear drift.

\item \textbf{Linear drift} is the least sensitive among the three,
requiring the largest amplitude to achieve fixed power.
\end{itemize}

These observations are consistent with the signal constants~\eqref{eq:canonical_A_values}

\[
A(g_{\mathrm{lin}}) = \frac{1}{4},
\quad
A(g_{\mathrm{sin}}) = \frac{1}{\pi},
\quad
A(g_{\mathrm{step}}) = \frac{1}{2},
\]

confirming that $A(g)$ captures the leading model dependence of the finite-
sample signal strength in the CUSUM statistic.

All simulations were implemented in Python with vectorized sampling routines
and fixed random seeds. Raw power tables and threshold estimates are stored as
structured data files for post-processing.

\subsection{Operational Implications}
\label{sec:discussion_block}

\begin{figure}[t]
\centering
\includegraphics[width=\columnwidth]{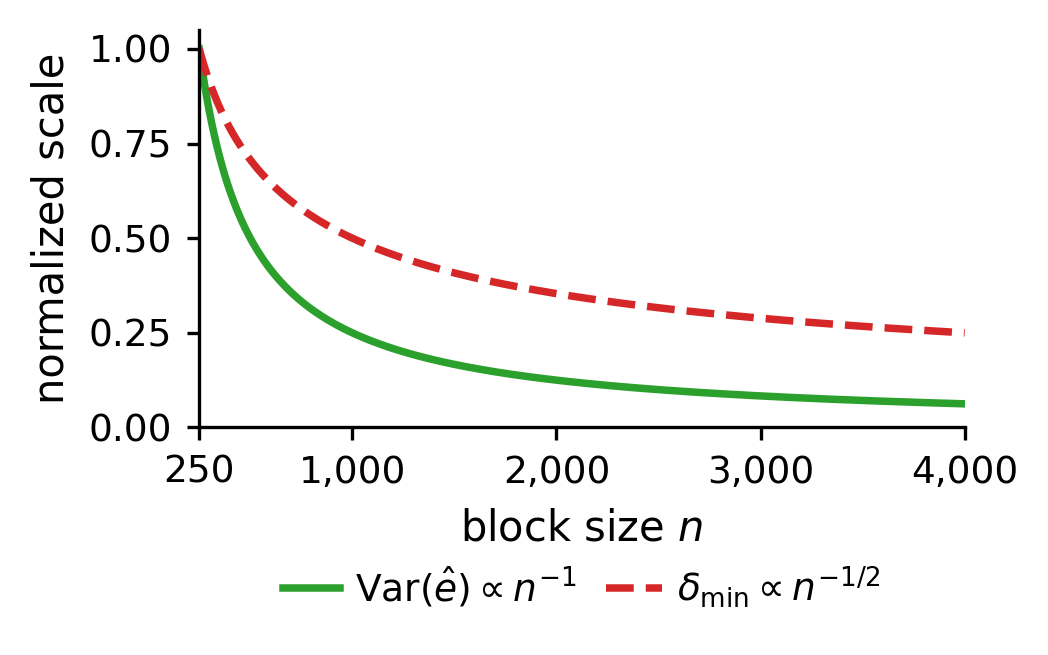}
\caption{Statistical trade-off imposed by block size $n$. As $n$ increases,
estimation variance $\mathrm{Var}(\hat e)=\sigma_0^2/n$ decreases (improving
key-rate stability), while the minimal detectable drift amplitude
$\delta_{\min}\propto n^{-1/2}$ also decreases --- but only at the
square-root rate. The two curves represent incommensurable gains: a system
operator who doubles $n$ to halve estimation noise also coarsens the temporal
resolution threshold by a factor of $\sqrt{2}$.}
\label{fig:tradeoff_bloco}
\end{figure}

The detectability rate
\[
\delta_{\min}(n,\alpha,\beta) = \Theta(n^{-1/2})
\]
imposes a statistical constraint on finite-key parameter monitoring.
Figure~\ref{fig:tradeoff_bloco} illustrates both scaling behaviors
simultaneously.

In finite-key QKD implementations, block size $n$ is chosen
to balance key-rate efficiency against statistical uncertainty.
Larger blocks reduce variance in parameter estimates and improve
key-rate stability.
The same increase in block size, however, reduces sensitivity to intra-block
temporal drift: within the single-block testing problem studied here, drifts
with amplitude $\delta = o(n^{-1/2})$ cannot be reliably distinguished from
stationary fluctuations at fixed confidence levels.

Formally, the variance of the empirical error rate satisfies
\[
\mathrm{Var}(\hat e) = \frac{\sigma_0^2}{n},
\]
while the minimal detectable structured deviation scales as
\[
\delta_{\min}(n,\alpha,\beta) \asymp \frac{1}{\sqrt{n}}.
\]
Thus increasing $n$ reduces estimation noise while improving drift resolution
only at the square-root rate: precision in the empirical rate $\hat e$
improves linearly in $n$, whereas the smallest resolvable structured drift
shrinks only as $n^{-1/2}$.

This limitation is statistical rather than algorithmic: it follows from the
vanishing Kullback--Leibler separation between the null hypothesis and
admissible alternatives whenever $n\delta^2 \to 0$. As emphasised in the
Introduction, this detectability boundary is complementary to composable
security certification rather than a constraint on it.

From a system-design perspective,
the block size implicitly determines both
key-rate stability and drift detectability.
If operational drift amplitudes are expected to scale
with environmental or hardware fluctuations,
then monitoring sensitivity must be evaluated
relative to the $n^{-1/2}$ statistical resolution threshold.

In particular, if expected environmental fluctuations
produce effective drift amplitudes of order $\delta_{\mathrm{phys}}$,
then detectability requires
\[
\delta_{\mathrm{phys}} \gtrsim n^{-1/2}.
\]
This condition provides an operational criterion
for selecting monitoring block sizes in deployed systems.

These findings motivate block partitioning strategies, adaptive monitoring
across blocks, or multi-scale diagnostics when slow drift is expected.
Such refinements can change constants and operational trade-offs, but they do
not alter the single-block square-root scaling derived here.

\subsection*{Synthesis}

Taken together, the simulations are consistent with the three predictions of
Section~\ref{sec:methods}: the $n^{-1/2}$ decay of $\widehat{\delta}_{\min}$,
the collapse of power curves under $\delta\sqrt{n}$, and the model-dependent
ordering reproduced quantitatively by the signal constants $A(g)$ with no
free parameters.

\section{Conclusion}
\label{sec:conclusion}

We studied intra-block drift detection in finite-key entanglement-based QKD
through a minimax hypothesis-testing framework. The null model is stationary,
while the alternative introduces mean-preserving temporal structure that is
invisible to global-average statistics.

The main conclusion is that the minimal detectable amplitude~\eqref{eq:delta_min_def} satisfies
\[
\delta_{\min}(n,\alpha,\beta) = \Theta(n^{-1/2}).
\]
The lower bound shows that when $n\delta^2 \to 0$, no level-$\alpha$ test can
guarantee nontrivial uniform power over the admissible drift class. The upper
bound shows that a calibrated CUSUM statistic detects drift at the matching
scale under the stated regularity conditions.

For finite-key QKD, this rate quantifies a monitoring-resolution limit rather
than a security limitation. Larger blocks stabilize parameter estimates, but
they also restrict the smallest intra-block drift that can be diagnosed from
finite data. Explicit signal constants for linear, sinusoidal, and step drift
profiles, together with the simulation results, confirm the predicted scaling
and its model dependence.

Future work may extend the analysis to dependent observations, multi-block
monitoring schemes, and adaptive diagnostics that remain compatible with
finite-key operational constraints.

From an operational standpoint, the result establishes that intra-block
temporal drift is not merely a theoretical concern: it defines a statistical
resolution limit that is present in any finite-key implementation, regardless
of the monitoring algorithm employed. Practitioners selecting block sizes
should treat $n^{-1/2}$ not only as a key-rate efficiency parameter, but as a
bound on their system's ability to diagnose channel instability from within a
single block.

\appendix
\section{Technical Derivations}

\subsection{Lower-Bound Proof Details}
\label{app:lower_bound_proof}

\begin{proof}[Proof of Theorem~\ref{thm:minimax_lower}]
Let $P_0$ denote the joint distribution under $H_0$ and $P_g$ the
distribution under a fixed alternative $g \in \mathcal{G}(L,c)$.

Because observations are independent,
\[
\mathrm{KL}(P_g \,\|\, P_0)
=
\sum_{i=1}^n
\mathrm{KL}\!\left(
\mathrm{Bern}(e_0+\delta g(t_i))
\,\big\|\,
\mathrm{Bern}(e_0)
\right).
\]

For Bernoulli parameters $p = e_0 + \varepsilon$ sufficiently close to $e_0$,
\[
\mathrm{KL}(\mathrm{Bern}(p)\|\mathrm{Bern}(e_0))
=
\frac{\varepsilon^2}{2 e_0(1-e_0)}
+ O(\varepsilon^3).
\]

With $\varepsilon = \delta_n g(t_i)$ and $\|g\|_\infty \le 1$,
\[
\mathrm{KL}(P_g \,\|\, P_0)
=
\frac{\delta_n^2}{2 e_0(1-e_0)}
\sum_{i=1}^n g(t_i)^2
+ O(n\delta_n^3).
\]

Because $g$ belongs to $\mathcal G(L,c)$ and is Lipschitz with constant $L$,
the Riemann-sum approximation yields
\[
\left|
\frac{1}{n}\sum_{i=1}^n g(t_i)^2
-
\int_0^1 g(t)^2 dt
\right|
\le
\frac{C(L)}{n}
\]
for some constant $C(L)>0$, hence
\[
\sum_{i=1}^n g(t_i)^2
=
n\int_0^1 g(t)^2dt + O(1).
\]
Since $\|g\|_2\ge c>0$, this preserves the $n\delta_n^2$ scaling.

Therefore, for sufficiently small $\delta_n$,
\[
\mathrm{KL}(P_g \,\|\, P_0)\le C n\delta_n^2
\]
for some constant $C>0$ depending on $e_0$ and $c$.

Now choose a finite subset $\{g_1,\dots,g_M\}\subset\mathcal G(L,c)$ and define
\[
P_\Pi=\frac{1}{M}\sum_{j=1}^M P_{g_j}.
\]
By convexity of KL divergence,
\[
\mathrm{KL}(P_\Pi \,\|\, P_0)
\le
\frac{1}{M}\sum_{j=1}^M \mathrm{KL}(P_{g_j}\|\!P_0)
\le
C n\delta_n^2.
\]
Pinsker’s inequality gives
\[
\mathrm{TV}(P_\Pi,P_0)^2
\le
\frac{1}{2}\mathrm{KL}(P_\Pi\|\!P_0)
\le
C'n\delta_n^2.
\]
If $n\delta_n^2\to 0$, then $\mathrm{TV}(P_\Pi,P_0)\to 0$.

For any test $\varphi$ of size at most $\alpha$,
\[
\frac{1}{M}\sum_{j=1}^M P_{g_j}(\varphi=1)
\le
\alpha + \mathrm{TV}(P_\Pi,P_0).
\]
Hence
\[
\inf_{1\le j\le M} P_{g_j}(\varphi=1)\le \alpha+o(1),
\]
which implies
\[
\inf_{g\in\mathcal G(L,c)}P_g(\varphi=1)\le \alpha+o(1).
\]
Applying this to $\varphi=\varphi_n$ proves the theorem.
\end{proof}

\subsection{Proof of Proposition~\ref{prop:cusum_upper}}
\label{app:cusum_upper_proof}

\begin{proof}[Proof of Proposition~\ref{prop:cusum_upper}]
\textbf{Step~1.}
The threshold is of order $O(1)$. Under $H_0$, the centred partial-sum
process $n^{-1/2}\sigma_0^{-1}S_{\lfloor nt\rfloor}$
converges weakly in $C[0,1]$ to a Brownian bridge by Donsker's
theorem~\cite{Billingsley1999}. Consequently
\[
T_n \xrightarrow{d} \sup_{t\in[0,1]}|B_0(t)|,
\]
where $B_0$ is a standard Brownian bridge. For each $\alpha\in(0,1)$ the
level-$\alpha$ quantile $q_\alpha := \sup\{x : P(\sup|B_0|>x)\ge\alpha\}$
is finite and depends only on $\alpha$. Thus $\tau_\alpha(n)=q_\alpha+o(1)$,
i.e.\ $\tau_\alpha(n)=O(1)$.

\textbf{Step~2 (signal under $H_1$).}
Under $H_1(\delta,g)$ the deterministic component of $S_k$ is
$\mu_k = \delta n G(k/n) + O(1)$, where $G(t)=\int_0^t g(u)\,du$~\cite{Basseville1993}.
The stochastic component remains of order $\sigma_0\sqrt{n}$. Therefore
\begin{align*}
T_n &= \frac{\max_k|S_k|}{\sqrt{n\hat{e}(1-\hat{e})}} \\
    &\ge \frac{\delta n A(g) + O(\sqrt{n})}{\sigma_0\sqrt{n}(1+o(1))} \\
    &= \frac{\delta\sqrt{n}\,A(g)}{\sigma_0} + O(1).
\end{align*}

\textbf{Step~3 (achievability).}
By the functional CLT, for any fixed $u>0$,
\begin{align*}
&\mathbb{P}_{H_1(\delta,g)}\!\bigl(T_n \le \tau_\alpha(n)\bigr) \\
&\quad\le
\mathbb{P}\!\left(Z \le \tau_\alpha(n) - \frac{\delta\sqrt{n}\,A(g)}{\sigma_0}\right)
+ o(1),
\end{align*}
where $Z$ is a standard normal random variable.
Setting the right-hand side equal to $\beta$ and solving for $\delta$ gives
\[
\delta \ge
\frac{\sigma_0\bigl(q_\alpha + \Phi^{-1}(1-\beta)\bigr)}{A(g)\sqrt{n}}
=: \frac{C_{\alpha,\beta}(e_0,A_\star)}{\sqrt{n}},
\]
where we used $A(g)\ge A_\star>0$.
Hence $\mathbb{P}_{H_1}(T_n>\tau_\alpha(n))\ge 1-\beta$ for all sufficiently
large $n$, completing the achievability argument~\cite{Lorden1971,Moustakides1986}.
\end{proof}

\subsection{Canonical Drift Derivations}
\label{app:canonical_derivations}

For completeness, we summarize the drift profiles used in the simulations:
\[
\begin{aligned}
\tilde g_{\mathrm{lin}}(t) &= 2t-1,\\
\tilde g_{\mathrm{sin}}(t) &= \sin(2\pi t),\\
\tilde g_{\mathrm{step}}(t) &=
\begin{cases}
-1,&0\le t<\tfrac{1}{2},\\
+1,&\tfrac{1}{2}\le t\le 1.
\end{cases}
\end{aligned}
\]
All are mean-preserving and satisfy $\|g\|_\infty=1$.

Their cumulative integrals are:
\[
G_{\mathrm{lin}}(t)=\int_0^t(2u-1)\,du=t^2-t,
\]
\[
G_{\mathrm{sin}}(t)=\int_0^t\sin(2\pi u)\,du
=\frac{1-\cos(2\pi t)}{2\pi},
\]
\[
G_{\mathrm{step}}(t)=
\begin{cases}
-t,&0\le t<\tfrac{1}{2},\\
t-1,&\tfrac{1}{2}\le t\le 1.
\end{cases}
\]

Therefore:
\[
\begin{aligned}
A(\tilde g_{\mathrm{lin}})&=\sup_t|G_{\mathrm{lin}}(t)|=\frac{1}{4},\\
A(\tilde g_{\mathrm{sin}})&=\sup_t|G_{\mathrm{sin}}(t)|=\frac{1}{\pi},\\
A(\tilde g_{\mathrm{step}})&=\sup_t|G_{\mathrm{step}}(t)|=\frac{1}{2}.
\end{aligned}
\]

For the step profile, smoothing near $t=\tfrac{1}{2}$ within the admissible
class preserves the same leading-order scaling conclusions.

\section{Simulation Details}
\label{app:simulation_details}

Threshold calibration uses $M_0 = 10{,}000$ Monte Carlo samples under
$H_0$ for each pair $(n, e_0)$. Power is estimated from $M_1 = 5{,}000$
samples under $H_1(\delta, g)$, yielding Monte Carlo standard errors of
approximately $\pm 0.01$ on all power estimates. The $\delta$ grid is
refined near the transition region $\delta \asymp n^{-1/2}$; empirical
thresholds $\widehat{\delta}_{\min}$ are obtained by linear interpolation
over this grid. Raw power tables and threshold estimates are stored as
structured data files for post-processing.

\backmatter{}

{
\sloppy
\bmhead{Acknowledgements}

The authors acknowledge institutional support from GWK.
\par}

\section*{Declarations}

\bmhead{Funding}

The authors declare that no funding was received for this research.

\bmhead{Competing interests}

The authors declare no competing interests.

\bmhead{Ethics approval and consent to participate}

Not applicable.

\bmhead{Consent for publication}

Not applicable.

\bmhead{Data availability}

The simulation data generated and analyzed in this study, including the
raw power tables and threshold estimates underlying all figures, are
available from the corresponding author on reasonable request.

\bmhead{Code availability}

The source code used to generate all simulations and figures is available
from the corresponding author on reasonable request.

\bmhead{Author contributions}

Following the CRediT taxonomy: Rafael Duarte Marcelino --- conceptualization,
methodology, formal analysis, software, validation, investigation,
visualization, writing (original draft).
Julio Smanioto Garcia --- software, validation, writing (review and
editing).
Matheus Rufino --- supervision, project administration, writing (review
and editing).
All authors read and approved the final manuscript.

\bibliography{references}

\end{document}